\documentclass[12pt]{article}
\usepackage{authblk}

\usepackage{float}
\usepackage{amsmath}
\usepackage{amsthm}
\usepackage{amssymb}
\usepackage{setspace}
\usepackage[authoryear]{natbib}
\usepackage[dvipdfmx]{graphicx}
\usepackage{color}
\usepackage{mathrsfs}
\usepackage{comment}

\usepackage[unicode=true,
 bookmarks=true,bookmarksnumbered=false,bookmarksopen=false,
 breaklinks=false,pdfborder={0 0 1},backref=page,colorlinks=true]
 {hyperref}
\hypersetup{citecolor=Emerald,  linkcolor=RoyalBlue}

\makeatletter

\floatstyle{ruled}
\newfloat{algorithm}{tbp}{loa}
\providecommand{\algorithmname}{Algorithm}
\floatname{algorithm}{\protect\algorithmname}

\numberwithin{equation}{section}
\numberwithin{figure}{section}
\theoremstyle{definition}

\theoremstyle{remark}
\newtheorem{rem}{\protect\remarkname}

\usepackage[dvipsnames]{xcolor}
\AtBeginDocument{%
\let\ref\autoref
\renewcommand\equationautorefname{\@gobble}
}
\providecommand{\examplename}{Example}
\providecommand{\remarkname}{Remark}
\newtheorem{prop}{Proposition}

\makeatother

\begin{document}
\title{
Adaptation of the Tuning Parameter in General Bayesian Inference with Robust Divergence  
}
\author[1,3]{Shouto Yonekura}
\author[2,3]{Shonosuke Sugasawa}
\affil[1]{Graduate School of Social Sciences, Chiba University}
\affil[2]{Center for Spatial Information Science, The University of Tokyo}
\affil[3]{Nospare Inc.}

\maketitle
\begin{abstract}
We introduce a novel methodology for robust Bayesian estimation with robust divergence (e.g., density power divergence or $\gamma$-divergence), indexed by tuning parameters. It is well known that the posterior density induced by robust divergence gives highly robust estimators against outliers if the tuning parameter is appropriately and carefully chosen. In a Bayesian framework, one way to find the optimal tuning parameter would be using evidence (marginal likelihood).
However, we theoretically and numerically illustrate that evidence induced by the density power divergence does not work to select the optimal tuning parameter since robust divergence is not regarded as a statistical model. To overcome the problems, we treat the exponential of robust divergence as an unnormalisable statistical model, and we estimate the tuning parameter by minimising the Hyvarinen score. We also provide adaptive computational methods based on sequential Monte Carlo (SMC) samplers, enabling us to obtain the optimal tuning parameter and samples from posterior distributions simultaneously. The empirical performance of the proposed method through simulations and an application to real data are also provided.
\end{abstract}


\newpage
\doublespacing

\section{Introduction}
\label{sec:Introduction}
One well-known way to deal with outliers and model misspecification when conducting inference is to use
robust divergences.
Since the pioneering work of \citet{basu1998robust} that proposed density power divergence as an extension of the standard likelihood, some variants of the divergence \citep[e.g.][]{fujisawa2008robust,cichocki2011generalized,ghosh2017generalized} and various statistical methods using robust divergences have been developed.
Many robust divergences are indexed by a single tuning parameter that controls the robustness against outliers. 
If the tuning parameter is set to a smaller value than necessary, the resulting estimator may still be affected by outliers. 
On the other hand, using an unnecessarily large value for the tuning parameter leads to a loss of statistical efficiency \citep{basu1998robust}.
Despite the success of the theoretical analysis of properties of statistical methods based on robust divergences, how to adaptively estimate the tuning parameter from the data has often been ignored, with a few exceptions \citep{warwick2005choosing,basak2021optimal} that propose frequentist methods to select the optimal value via asymptotic mean squared errors. 
There is a growing body of literature on Bayesian approaches using robust divergence.
For example, general theory has been considered by \citep{ghosh2016robust,jewson2018principles,nakagawa2020robust} and some specific applications to statistical models such as linear regression \citep{hashimoto2020robust}, change point detection \citep{knoblauch2018doubly} and Bayesian filtering \citep{boustati2020generalized}.
Nevertheless, a reasonable estimation strategy for the tuning parameter has not been carefully discussed. 
A natural consideration to find the best tuning parameter in the context of Bayesian statistics will be the use of model evidence or marginal likelihood. 
However, as we shall illustrate later, evidence is not useful for choosing the tuning parameter since the exponential of robust divergence cannot be directly interpreted as a normalised statistical model.

In this paper, this issue is addressed by taking advantage of ideas from statistical theories for unnormalised statistical models \citep{hyvarinen2005estimation} introducing the Hyvarinen score (H score), which is a finite version of Fisher divergence.
Based upon the idea of \citet{hyvarinen2005estimation}, \citet{dawid2015bayesian,shao2019bayesian} consider unnormalisable marginal likelihoods, with particular attention to model selection, where such unnormalisability is driven by the improper prior.
Our main idea is to regard the exponential of robust divergence as unnormalisable models and employ a posterior distribution driven by the H-score, inspired by \citet{dawid2015bayesian, shao2019bayesian}, as an objective function of $\gamma$.
Since our objective function cannot be computed analytically in general, we then take advantage of sequential Monte Carlo (SMC) samplers \citep{del2006sequential} within a Robbins-Monro stochastic gradient framework to approximate the objective function, to estimate the optimal tuning parameter and to obtain samples from posterior distributions.

Therefore, our work can be understood as an attempt to fill the current gap between the existing theory of such robust estimation and their practical implementation. Our proposed method has the following advantages over existing studies.
\begin{enumerate}
    \item Unlike existing methods \citep{warwick2005choosing,basak2021optimal}, our proposed method does not require a pilot plot.
     To optimise a tuning parameter, it is necessary to determine a certain value as a pilot estimate. 
     Thus, the estimates may be strongly dependent on the pilot estimate. In addition, such methods often estimate excessively large values of the tuning parameters, given the proportion of outliers in the data. 
     In contrast, our algorithm is stable and statistically efficient since that does not require a pilot estimate.
    \item Proposed methods in \citep{warwick2005choosing,basak2021optimal} require an analytical representation of the asymptotic variance that cannot be obtained in general.
    Compared to such methods, our proposed method does not require such expression, and therefore our method can be applied to rather complex statistical models, which seem difficult to be handled by the previous methods.
    \item We take advantage of SMC samplers \citep{del2006sequential} in the generalised Bayesian framework \citep{bissiri2016general} with a gradient-ascent approach to perform parameter inference for the tuning parameter and posterior sampling simultaneously.
    This is a unique favourable algorithmic characteristic compared to methods that estimate parameters by fixing tuning parameters or by estimating tuning parameters once and then estimating other objects of interest.
\end{enumerate}

Recently, \cite{jewson2021general} has introduced a new Bayesian framework called {\it H-posterior} for unnormalisable statistical models based on Fisher divergence, and they have developed model selection criterion via the Laplace approximation of the marginal likelihood. 
The biggest difference from \cite{jewson2021general} is that we use a natural form of general posterior based on robust divergence, which is widely adopted in the literature \citep[e.g.][]{ghosh2016robust,jewson2018principles}, while the form of the posterior distribution in \cite{jewson2021general} is different from ours. 
As we mentioned, the main contribution of their work is the construction of model selection criteria of the BIC type through the Laplace approximation and the proof of their consistency.
On the other hand, our research is about the estimation of tuning parameters and the inference of the posterior distribution, and the main objective is to propose an objective function and a computational method considered suitable for this purpose.
In the framework of Generalised Bayesian, some methods that use variational Bayesian inference \citep{knoblauch2019generalized,frazier2021loss} have also been proposed.
However, instead of computational speed, the approximate distribution obtained does not match the target distribution in the limit, and the method of estimating the tuning parameters is unclear. 
There is also a need to make some natural but somewhat stronger assumptions than our proposed method, such as that the target distribution is an exponential family.

The rest of the paper is organised as follows. 
In \ref{sec:robustBayse}, we first set up the framework and then show theoretically and numerically that evidence induced by density power divergence to select the tuning parameter. Instead, we propose to estimate it based on the H-score \citep{hyvarinen2005estimation,dawid2015bayesian} and characterise its asymptotic behaviour. As mentioned earlier, our method involves functions for which it is difficult to obtain an analytic representation. Therefore, we develop an adaptive and efficient Markov chain Monte
Carlo (MCMC) algorithm based on SMC samplers \citep{del2006sequential} in \ref{sec:SMCsampler}. Numerical applications of the developed methodology are provided in \ref{sec:Numerical-examples}, then conclusions and directions for future research are provided in \ref{sec:coclusion}.

\section{Bayesian Inference with Robust Divergence } 
\label{sec:robustBayse}
\subsection{General posterior distribution}

Suppose that we have $d_{y}$-dimensional $i.i.d.$ data $\{y_{i}\}_{i=1}^{n}\overset{i.i.d.}{\sim}G$
where $G$ denotes the true distribution or the data-generating process.
Also, assume that we have a statistical model $\{f_{\theta}:\theta\in\Theta\}$
where $\Theta\subseteq\mathbb{R}^{d}$ for some $d\geq1.$ We then
write $y_{1:n}$ to denote $(y_{1},\ldots,y_{n}$) and let $g$
denote the density of $G$ with respect to $dy$. 
To make robust Bayesian inferences for $\theta$, we use a potential function based on robust divergence instead of the standard likelihood function.

Here, we simply consider \emph{the density power divergence} \citep{basu1998robust} but other ones with tuning parameters, such as $\gamma$-divergence \citep{fujisawa2008robust, nakagawa2020robust}, $\alpha$-divergence \citep{cichocki2010families}, H\"older divergence \citep{nielsen2017holder}, also can be used within our framework. 
Given a prior density
$\pi(\theta)$ with respect to $d\theta$, we can define the corresponding posterior density
\begin{align}
\text{\ensuremath{\Pi}}_{\gamma}(\theta\mid y_{1:n}):= & \frac{\mathcal{L}_{\gamma}(y_{1:n};\theta)\pi(\theta)}{p_{\gamma}(y_{1:n})},\label{eq:posterior}
\end{align}
where $\gamma\in(0,\infty]$, $p_{\gamma}(y_{1:n}):=\int_{\Theta}\mathcal{L}_{\gamma}(y_{1:n};\theta)\pi(\theta)d\theta$
and 
\begin{align}
\log\mathcal{L}_{\gamma}(y_{1:n};\theta) & :=\sum_{i=1}^n\log\mathcal{L}_{\gamma}(y_i;\theta), \ \ \ \ 
\log\mathcal{L}_{\gamma}(y_i;\theta):=\frac{1}{\gamma}f_{\theta}(y_{i})^{\gamma}-\frac{1}{1+\gamma}\int f_{\theta}(x)^{1+\gamma}dx.
\label{eq:Beta_likeli-1}
\end{align}

Note that $p_{\gamma}(y_{1:n})$ is generally referred to as \emph{evidence}.
In many scenarios, robustified posterior densities such as \eqref{eq:posterior} give much more accurate and stable inference against outliers and theoretical properties of the posterior have been investigated in \citep{ghosh2016robust,nakagawa2020robust}.
However, its performance depends critically on the choice of the tuning parameter $\gamma$ in \eqref{eq:posterior} \citep[e.g.][]{ghosh2016robust}, which motivate us to find \textquotedblleft best\textquotedblright{}
$\gamma$ to make inference successful. 
Notice that \eqref{eq:posterior} can be seen as a special case of general Bayesian updating \citep{bissiri2016general} with weight setting $1$.
As noted in \cite{jewson2018principles}, the density power divergence does not have any arbitrariness in the scale as a loss function, and one can set $\omega=1$.
Under the general framework of Bayesian updating, Corollary 1 of \citet{fong2020marginal} implies that evidence is still the unique coherent marginal score for Bayesian inference.
Thus, from the viewpoint of Bayesian statistics, it appears to be natural to find the best $\gamma$ based on evidence, but its property of it is unclear since $\mathcal{L}_{\gamma}(y_{1:n};\theta)$ is not a probability density of $y_{1:n}$. 
Furthermore, the tuning parameter $\gamma$ cannot be interpreted as \textquotedblleft model parameter\textquotedblright{} in this case. 
The following example highlights the problem of using $p_{\gamma}(y_{1:n})$ to find the best $\gamma$.

\subsection{Failure of model evidence: motivating example}\label{egxamle:toyexample} 
To see why evidence is not useful for estimating $\gamma$, we start with the following proposition for a rescaled $\log\mathcal{L}_{\gamma}(y_{i};\theta)$. 

\begin{prop} \label{prop:monotones}
Consider $\log\mathcal{L}^{\mathbf{R}}_{\gamma}(y_{i};\theta):=\log\mathcal{L}_{\gamma}(y_{i};\theta)-\frac{1}{\gamma}+1$. Furthermore, assume that $f(x)\leq1$ for any $x$. 
Then $\log\mathcal{L}^{\mathbf{R}}_{\gamma}(y_{i};\theta)$ is a monotonically increasing function of $\gamma$.

\end{prop}
\begin{proof}
See \ref{sec:proofmonotone}.
\end{proof}
Since the term $-\frac{1}{\gamma}+1$ is eliminated when considering the posterior distribution (\ref{eq:posterior}), this rescaling is a non-essential modification in the method we shall propose later. 
The meaning of the rescaling is to ensure that $\log\mathcal{L}^{\mathbf{R}}_{\gamma}(y_{i};\theta)$ converges to the log-likelihood as $\gamma\to 0$, so that $\log\mathcal{L}^{\mathbf{R}}_{\gamma}(y_{i};\theta)$ can be regarded as a natural extension of the log-likelihood.

The important point here is that Proposition \ref{prop:monotones} implies that there are theoretically at least some situations where evidence is increasing monotonically for $\gamma$. 
Indeed, the following numerical example vividly illustrates such a situation.
To see this numerically, we consider a simple but motivating example in which $\{y_{i}\}_{i=1}^{100}\overset{i.i.d.}{\sim}G=\mathcal{N}(1,1)$
and then randomly replace $\tau\%$ of $\{y_{i}\}_{i=1}^{100}$ by $y_{i}+5$, where $0\leq\tau\leq100$ is called the contamination proportion. 
Here $\gamma$ was determined by dividing equally $[0.01,1]$ into 1,000 points. In other words, in the context of Bayesian model selection, this corresponds to choosing the model with the largest evidence as to the best model out of 1,000 models indexed by $\gamma$. 
With the choice $\tau=10$, we then calculated $2,000$ Monte Carlo estimates of the model evidence $p_{\gamma}(y_{1:n})$
for each $\gamma$. The resulting $p_{\gamma}(y_{1:n})$ are shown in \ref{fig:Example1}, which numerically shows that $p_{\gamma}(y_{1:n})$ is a monotonically increasing function of $\gamma$ so that it does not have local maxima. 
This implies that one cannot estimate $\gamma$ using $p_{\gamma}(y_{1:n})$.
A similar phenomenon is also discussed in \cite{jewson2021general}.

\begin{figure}[H]
\begin{centering}
\includegraphics[scale=0.6]{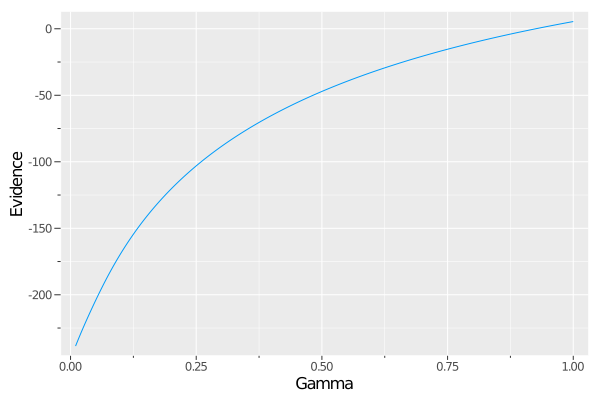}
\par\end{centering}
\vspace{-0.3cm}
\caption{Estimated $p_{\gamma}(y_{1:n})$. The Y-axis represents the value of $p_{\gamma}(y_{1:n})$, and the X-axis represents the value of $\gamma$.}
\label{fig:Example1}
\end{figure}

\subsection{Estimation using H-score}
To overcome the illustrated problem, we first treat $\mathcal{L}_{\gamma}(y_{1:n};\theta)$ as an unnormaliseable statistical model motivated by \citet{hyvarinen2005estimation}.
Note that even with an unnormalisable model, the update in \eqref{eq:posterior} can be considered as a valid belief update according to \cite{bissiri2016general}.
It should be noted that when $\log\mathcal{L}_{\gamma}(y_{1:n};\theta)$ is the density power divergence of the form (\ref{eq:Beta_likeli-1}), the normalising constant may not exist.
For example, when $f_{\theta}$ is a normal distribution, $\log\mathcal{L}_{\gamma}(y_i;\theta)$ converges to a constant value under $|y_i|\to \infty$, so the integral of $\mathcal{L}_{\gamma}(y_{1:n};\theta)$ with respect to $y_{1:n}$ diverges.   
Recently, \cite{jewson2021general} has pointed out that the role of such unnormalisable models can be recognised in terms of relative probability. 

For $d_y$ dimensional observations $y$ and twice differentiable density $p(\cdot)$, \citet{hyvarinen2005estimation} defines the H-score as
\begin{equation*}
    \mathcal{H}(y,p):=\sum_{k=1}^{d_y}\left\{2\frac{\partial^2\log p(y)}{\partial y_{(k)}^2}+\left( \frac{\partial \log p(y)}{\partial y_{(k)}}\right)^2\right\}.
\end{equation*}
We then select the optimal $\gamma$ with the smallest leave-one-out H-score, defined as 
\begin{equation}
    \sum_{i=1}^n \mathcal{H}(y_i,p_{\gamma}(y_i|y_{-i})), \label{eq:Hscore}
\end{equation}
where $p_{\gamma}(y_i|y_{-i})=\int \mathcal{L}_{\gamma}(y_{i};\theta)\Pi_{\gamma}(\theta|y_{-i})d\theta$ and $y_{-i}=(y_1,...,y_{i-1},y_{i+1},...,y_n)$.
Note that \citet{shao2019bayesian} adoptes the H-score to define prequential score for state space models, and the criteria (\ref{eq:Hscore}) can be seen as prequential score under {\it i.i.d.} settings. 
As shown in \ref{sec:derivation}, under the assumptions stated in \citet{shao2019bayesian}, the leave-one-out H-score \eqref{eq:Hscore} can be rewritten as 
\begin{equation}
\mathcal{H}_{n}(\gamma)
:=\sum_{i=1}^{n}\sum_{k=1}^{d_y}\left\{2\mathbb{E}_{\gamma}\left[\frac{\partial^{2}\log\mathcal{L}_{\gamma}(y_i;\theta)}{\partial^{2}y_{i_{(k)}}}+\left(\frac{\partial\log\mathcal{L}_{\gamma}(y_{i};\theta)}{\partial y_{i_{(k)}}}\right)^{2}\right]-\left(\mathbb{E}_{\gamma}\left[\frac{\partial\log\mathcal{L}_{\gamma}(y_{i};\theta)}{\partial y_{i_{(k)}}}\right]\right)^{2}\right\},
\label{eq:H}
\end{equation}
where the expectation is with respect to the robustified posterior distribution \eqref{eq:posterior}. 
Then, we can estimate $\gamma$ as follows
\begin{equation}
\hat{\gamma}:=\arg\min_{\gamma}\mathcal{H}_{n}(\gamma)
\label{eq:gam_hat}.
\end{equation}
As we shall discuss later, it can be shown that, under some conditions, $n^{-1}\mathcal{H}_{n}(\gamma)$ converges to the Fisher divergence,
$\mathcal{J}(\gamma):=\int\left\Vert \nabla_{y}\log g(y_{1:n})-\nabla_{y}\log p_{\gamma}(y_{1:n})\right\Vert ^{2}g(y_{1:n})dy_{1:n}$
Therefore, $\mathcal{H}_{n}(\gamma)$ can be considered as an empirical approximation of the Fisher divergence $\mathcal{J}(\gamma)$ for the marginal distribution based on unnormaliseable models defined by robust divergence. 
An important point here is that the estimation by the H-score is independent of the normalisation constant.
The following proposition is theoretical justification of selecting $\gamma$ via (\ref{eq:gam_hat}).

\begin{prop} \label{pro:consistency}
\emph{Let $\gamma^{\star}:=\arg\min_{\gamma}\mathcal{J}(\gamma)$. Then, under the conditions stated in \ref{sec:proofconsistency}, we have $\hat{\gamma}\rightarrow\gamma^{\star}$ w.p.1. as $n\rightarrow\infty$.}
\end{prop}

\begin{proof}
See \ref{sec:proofconsistency}.
\end{proof}

\begin{rem}
\emph{As we mentioned, the prequential version of $\mathcal{H}_{n}(\gamma)$ is also
called the H-score in the context of Bayesian model selection \citep{shao2019bayesian,dawid2015bayesian}.
The main advantage of using the H-score in this context is that it will provide a consistent and coherent model selection criterion. 
\citet{jewson2021general} proposes a consistent model selection criterion that is similarly based on H-scores but with batch estimation.
Although the prequential method is coherent, this comes with a very high computational cost, for every model, one must do posterior inference on all permutations of the data and increasing sample sizes.
Here, we use a batch estimation approach to estimate $\gamma$, which avoids high computational costs.
We also want to emphasise that, as we shall study later, such a batch approach will give rise to natural and efficient algorithms to estimate $\gamma$ and posterior sampling.
}
\end{rem}

\bigskip
Under the density power divergence (\ref{eq:Beta_likeli-1}), the first and second order derivatives of $\log\mathcal{L}_{\gamma}(y_i;\theta)$ are given by 
\begin{equation*}
\begin{split}
\frac{\partial \log\mathcal{L}_{\gamma}(y_i;\theta)}{\partial y_i}
&=f_{\theta}(y_i)^{\gamma-1}\frac{\partial f_{\theta}(y_{i})}{\partial y_i},\\
\frac{\partial^2 \log\mathcal{L}_{\gamma}(y_i;\theta)}{\partial y_i^2}
&=(\gamma-1)f_{\theta}(y_i)^{\gamma-2}\left(\frac{\partial f_{\theta}(y_{i})}{\partial y_i}\right)^2
+f_{\theta}(y_i)^{\gamma-1}\frac{\partial^2 f_{\theta}(y_{i})}{\partial y_i^2}.
\end{split}
\end{equation*}
These expressions do not include the integral term $\int f_{\theta}(x)^{1+\gamma}dx$, which makes the calculation of $\mathcal{H}_{n}(\gamma)$ much more straightforward in practice since the integral term often is a form of a complicated expression.

\subsection{Numerical illustration of the H-score under normal distribution}\label{Example:Hscore}
We consider the same problem in the example in \ref{egxamle:toyexample}. 
For a normal distribution $\mathcal{N}(\mu, \sigma^2)$, the derivatives are as follows
\begin{gather*}
\frac{\partial \log\mathcal{L}_{\gamma}(y_i;\theta)}{\partial y_i}
=
-\frac{\phi(y_i;\mu,\sigma^2)^{\gamma}(y_i-\mu)}{\sigma^2}, \ \ \ \ 
\frac{\partial^2 \log\mathcal{L}_{\gamma}(y_i;\theta)}{\partial y_i^2}
=\frac{\phi(y_i;\mu,\sigma^2)^{\gamma}}{\sigma^4}\left\{\gamma(y_i-\mu)^2 -\sigma^2\right\},
\end{gather*}
where $\phi(\cdot;\mu, \sigma^2)$ is the density function of $\mathcal{N}(\mu, \sigma^2)$.
We calculated $\mathcal{H}_n(\gamma)$ in (\ref{eq:H}) for each $\gamma$, where posterior expectations were approximated by $2000$ posterior samples of $\mu$.
The data were simulated in the same way as in \ref{egxamle:toyexample}.
The results are shown in \ref{fig:Example2} when $\tau=10$ (blue lines) and $30$ (red lines).
Our experiment shows numerically that $\mathcal{H}_n(\gamma)$ has a local minimum.
Furthermore, it can be seen that in regions where $\gamma$ is small, the posterior mean is relatively heavily influenced by outliers. In contrast, the posterior mean settles to a constant value in regions where $\gamma$ is greater than the value that minimises $\mathcal{H}_n(\gamma)$. Uncertainty in the sense of CI becomes greater.
This result would suggest that statistical inefficiencies occur in regions where $\gamma$ is larger than is necessary.

\begin{figure}[H]
\begin{centering}
\includegraphics[scale=0.7]{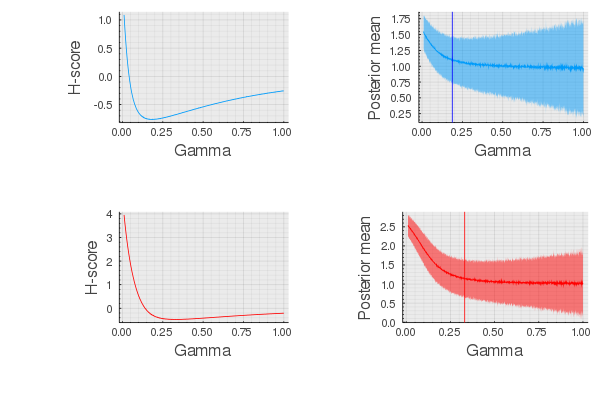}
\par\end{centering}
\caption{The top left-hand plots the values of $\mathcal{H}_{n}(\gamma)$ on the  Y-axis, and values of  $\gamma$ on the X-axis when $\tau=10$.
The minimum value of $\mathcal{H}_{n}(\gamma)$ was obtained when $\gamma=0.1874$.
The top right-hand figure plots the sample mean values of the posterior mean of $\mu$ on the Y-axis and the corresponding values of $\gamma$ on the X-axis under the same setting.
The vertical line represents the value of $\gamma$ that minimises $\mathcal{H}_{n}(\gamma)$, and the thin ribbon line represents the $95\%$ credible interval.
The bottom left-hand and right-hand figures represent the same figures when $\tau=30$ respectively.
In this case, he minimum value of $\mathcal{H}_{n}(\gamma)$ was obtained when $\gamma=0.3311.$
}
\label{fig:Example2}
\end{figure}

\section{Sequential Monte Carlo Samplers} \label{sec:SMCsampler}

A natural way to obtain $\hat{\gamma}$ in \eqref{eq:gam_hat} will
be to use Robbins-Monro-type recursion.
\begin{align}
\gamma_{t+1} & =\gamma_{t}+\kappa_{t}\nabla_{\gamma}\mathcal{H}_{t}(\gamma_{t}),\label{eq:Robbins-Monrro}
\end{align}
where $\sum_{t}\kappa_{t}=\infty,\sum_{t}\kappa_{t}^{2}<\infty$ but,
in general, posterior sampling based on the Monte Carlo approximation
will be required to evaluate $\nabla_{\gamma}\mathcal{H}_{t}(\gamma_{t})$.
That is, we need to construct an estimator of $\nabla_{\gamma}\mathcal{H}_{t}(\gamma_{t})$.
To do so, we first treat $\gamma_{t}$ in \eqref{eq:Robbins-Monrro}
as the positive sequence such that $0<\gamma_{0}<\gamma_{1}\cdots<\cdots\gamma_{T}$
where $0\leq t\leq T$ is an artificial time index. 
Then \eqref{eq:posterior} gives rise to the following tempering-like distributions on a common measurable space, say $(\Theta,\mathcal{B}(\Theta))$
\begin{align}
\text{\ensuremath{\Pi}}_{t}(\theta\mid y_{1:n}) & :=\text{\ensuremath{\Pi}}_{\gamma_{t}}(\theta\mid y_{1:n})\propto\mathcal{L}_{\gamma_{t}}(y_{1:n};\theta)\pi(\theta).\label{eq:tempred' distribuions}
\end{align}
The meaning of tempering-like here is not the usual tempering.
It means that a family of distributions is constructed in the same measurable space by a sequence of updated $\gamma_t$, gradually approaching the target distribution in the sense of an optimised $\gamma$, say $\gamma^{\star}$.
Using this, we can define a sequence of distributions defined on product
spaces $(\Theta^t,\mathcal{B}(\Theta^t):=(\Theta^{t}:=\prod_{i=1}^{t}\Theta)$
\begin{align}
\tilde{\Pi}_{t}(\theta_{0:t}\mid y_{1:n}) & :=\text{\ensuremath{\Pi}}_{t}(\theta_{t}\mid y_{1:n})\prod_{k=0}^{t-1}L_{k}(\theta_{k+1},\theta_{k}),\label{eq:istributions defined on product spaces}
\end{align}
where $L_{k}$ is a transition kernel from $\Theta^{k+1}$ to $\Theta^{k}$. 
Notice that $\tilde{\Pi}_{t}(\theta_{0:t}\mid y_{1:n})$ admits marginally $\ensuremath{\Pi}_{t}(\theta_{t}\mid y_{1:n})$.
Also let $M_{k}(\theta_{k-1},\theta_{k})$ be a $\text{\ensuremath{\Pi}}_{k}-$reversible MCMC kernel. 
Then it is given by 
$$
L_{k}(\theta_{k+1},\theta_{k})=\frac{\text{\ensuremath{\Pi}}_{k}(\theta_{k-1}\mid y_{1:n})M_{k}(\theta_{k-1},\theta_{k})}{\text{\ensuremath{\Pi}}_{k}(\theta_{k}\mid y_{1:n})}.
$$
Since $L_{t-1}\otimes \Pi_t=\Pi_t \otimes M_t$ by construction, as the Radon–Nikodym derivative between them, one can derive unnormalised incremental weights as follows
\begin{align}
\log w_{t}^{j} & :=\log\left(\frac{\mathcal{L}_{\gamma_{t}}(y_{1:n};\theta_{t-1}^{j})}{\mathcal{L}_{\gamma_{t-1}}(y_{1:n};\theta_{t-1}^{j})}\right)=\frac{-\frac{n}{1+\gamma_{t}}\int f_{\theta_{t-1}^{j}}(x)^{1+\gamma_{t}}dx+\frac{1}{\gamma_{t}}\sum_{i=1}^{n}f_{\theta_{t-1}^{j}}(y_{i})^{\gamma_{t}}}{-\frac{n}{1+\gamma_{t-1}}\int f_{\theta_{t-1}^{j}}(x)^{1+\gamma_{t-1}}dx+\frac{1}{\gamma_{t-1}}\sum_{i=1}^{n}f_{\theta_{t-1}^{j}}(y_{i})^{\gamma_{t-1}}}.\label{eq:unnormalized incremental weight}
\end{align}
For a detailed discussion of the choice of $L_k$ and $M_k$ and how the weights are derived, see e.g. \citet{del2006sequential,dai2020invitation}.
Then the SMC samplers \citep{del2006sequential} iterate the following steps. 
First, the normalised weights $W_{t}^{j}:=(\sum_{k=1}^{N}w_{t}^{k})^{-1}w_{t}^{j}$ are calculated for $j\in[1,N]$. 
Using these, one needs to sample ancestor indices $\left\{ A_{t}^{j}\right\} _{j=1}^{N}$ from the categorical distribution induced by the normalised weights $\{W_{t}^{j}\}_{j=1}^{N}$, denoted by $\mathrm{C}\left(\{W_{t}^{j}\}_{j=1}^{N}\right)$.
Finally, sample $\theta_{t}^{j}$ through $M_{t}(\theta_{t-1}^{A_{t}^{j}},d\theta_{t})$
for $j\in[1,N]$. We refer to \citet{dai2020invitation} for a number
of recent advances in SMC samplers. As a result, we have the particle system $\{\theta_{t}^{j},W_{t}^{j}\}_{j=1}^{N}$
that constructs an approximation $\sum_{j=1}^{N}W_{t}^{j}\delta_{\theta_{t}^{j}}(d\theta)$
of $\text{\ensuremath{\Pi}}_{t}(\theta\mid y_{1:n})$ at any time step $t$ as required,
where $\delta_{z}(dx)$ denotes the Dirac measure located at $z$. Using the system, an approximation $\nabla_{\gamma}\hat{\mathcal{H}}_{t}(\gamma_{t})$ of $\nabla_{\gamma}\mathcal{H}_{t}(\gamma)$ can be obtained under appropriate regularity conditions that allow us to interchange differentiation with respect to $\gamma$ and integration; see \ref{sec:derivative_Hscore} for details. Therefore, an approximation of \eqref{eq:Robbins-Monrro} will be
\begin{equation}
\gamma_{t+1}=\gamma_{t}+\kappa_{t}\nabla_{\gamma}\hat{\mathcal{H}}_{t}(\gamma_{t}).\label{eq:updatebeta}
\end{equation}
Given the time steps $T>0$ and initial $\gamma_{0}>0$, we update
$\gamma_{t}$ via \eqref{eq:updatebeta} and iterate the SMC sampler
$T$ times. Our method can be algorithmically summarised as follows.

\begin{algorithm}[H]
\caption{\label{alg:SMC-sampler-.}}

\begin{enumerate}
\item Initialise the particles $\{\theta_{0}^{j},W_{0}^{j}\}_{j=1}^{N}$,
with $\theta_{0}^{j}\overset{i.i.d.}{\sim}\pi$, $\gamma_{0}>0$,
$W_{0}^{i}=1$ for $j\in[1,N]$.
\item Apply the iteration $\gamma_{t+1}=\gamma_{t}+\kappa_{t}\nabla_{\gamma}\hat{\mathcal{H}}_{t}(\gamma_{t}).$
\item Calculate $w_{t}^{j}$ in \eqref{eq:unnormalized incremental weight}
and set $W_{t}^{j}=\frac{w_{t}^{j}}{\sum_{k=1}^{N}w_{t}^{k}}$ for
$j\in[1,N]$.
\item Sample ancestor indices $\left\{ A_{t}^{j}\right\} _{j=1}^{N}\sim\mathrm{C}\left(\{W_{t}^{j}\}_{j=1}^{N}\right)$.
\item Sample particles $\theta_{t}^{j}\sim M_{t}(\theta_{t-1}^{A_{t}^{j}},d\theta_{t})$
for $j\in[1,N]$.
\item Obtain estimate of $\nabla_{\gamma}\mathcal{H}_{t+1}(\gamma_{t+1})$.
\end{enumerate}
\end{algorithm}

\begin{rem}
\emph{Since $\{w_{t}^{j}\}$ are independent of $\{\theta_{t}^{i}\}$
but dependent of $\{\theta_{t-1}^{i}\}$, the particles $\{\theta_{t}^{i}\}$
can be sampled after resampling in \ref{alg:SMC-sampler-.}. In addition,
\ref{alg:SMC-sampler-.} uses a simple multinomial resampling applied
at each step. The variability of the Monte Carlo estimates can be
further reduced by incorporating dynamic resampling via the use of effective
sample size. See \citet{del2006sequential,dai2020invitation} for details.
}
\end{rem}

Theoretical guarantees of convergence of the Robbins-Monro algorithm usually require that $\nabla_{\gamma}\hat{\mathcal{H}}_{t}(\gamma_{t})$ is unbiased.
Even if $0<\gamma_{0}<\gamma_{1}\cdots<\cdots\gamma_{T}$ is chosen adaptively,  \cite{beskos2016convergence} shows that $\nabla_{\gamma}\hat{\mathcal{H}}_{t}(\gamma_{t})$ is still a (weakly) consistent estimator, but not an unbiased estimator.
Such unbiased estimation may be possible by using the recently developed MCMC with couplings in \ref{alg:SMC-sampler-.}, see \cite{middleton2019unbiased,jacob2020unbiased} for details. 
Instead of discussing convergence through such unbiased estimation, we shall discuss convergence through numerical experiments in the following sections.

We end this section by noting several advantages of the proposed method. First, \ref{alg:SMC-sampler-.} enables us to estimate the tuning parameter and obtain posterior sampling simultaneously. 
This is a notable difference from existing methods, such as running MCMC or the EM algorithm with a fixed tuning parameter, for example, \citet{fujisawa2008robust,ghosh2016robust}. 
We believe that it may be emphasised that by setting up a well-defined target function and using the stochastic gradient framework-based SMC samplers, it is possible to avoid the two-stage estimation that many previous studies have done in this context.
We also emphasise that our proposed method has two notable advantages over existing methods: it does not require pilot plots, and it does not require an expression of the asymptotic variance of the model. 
Next, recall that as $\gamma\downarrow0$, $\mathcal{L}_{\gamma}(y_{1:n};\theta)$ converges to Kullback–Leibler divergence. 
Let $\gamma^{\star}$ be the value of converged $\gamma$ in \ref{alg:SMC-sampler-.}.
Then \ref{alg:SMC-sampler-.} may be producing an approximated bridge between (multiplied by a prior distribution) Kullback–Leibler divergence and the target distribution $\mathcal{L}_{\gamma^{\star}}(y_{1:n};\theta)$.
Therefore, sampling from such tempering-like distributions induced by the density power divergence \eqref{eq:tempred' distribuions} could provide a beneficial tempering effect and a potential reduction in computational complexity, particularly when $d$ is large \citep{neal2001annealed}. 
Finally, \ref{alg:SMC-sampler-.} will give rise to a natural way to construct adaptive MCMC kernels. Suppose that we use Metropolis-Hastings kernels based on a Gaussian random walk. Notice that before MCMC step in \ref{alg:SMC-sampler-.}, we have $\{\theta_{t-1}^{j},W_{t-1}^{j}\}_{j=1}^{N}$ which approximates
$\text{\ensuremath{\Pi}}_{t-1}(\theta\mid y_{1:n})$ so that estimates $\hat{\mu}_{t-1}:=\sum_{j=1}^{N}W_{t-1}^{j}\theta_{t-1}^{j}$, $\hat{\Sigma}_{t-1}:=\sum_{j=1}^{N}W_{t-1}^{j}(\theta_{t-1}^{j}-\hat{\mu}_{t-1})(\theta_{t-1}^{j}-\hat{\mu}_{t-1})^{\top}$ are available. Also \citet[Theorem 1]{ghosh2016robust} shows that
$\text{\ensuremath{\Pi}}_{t-1}(\theta\mid y_{1:n})$ can be approximated by the Gaussian distribution. These will lead us to set $M_{t}(\theta_{t-1}^{A_{t}^{j}},d\theta_{t})=\theta_{t-1}^{A_{t}^{j}}+\xi_{t}$,
$\xi_{t}\overset{i.i.d.}{\sim}\mathcal{N}(0,2.38d^{-1/2}\hat{\Sigma}_{t-1})$, for instance. We note that this proposal is also can be considered as a consequence of the results from the optimal scaling analysis for random walk Metoropolis, see \citet[Chapter 17]{chopin2002sequential,chopin2020introduction}
and references therein for more details.

\section{Numerical Examples}\label{sec:Numerical-examples}
To specify the schedule for the scaling parameters $\{\kappa_{t}\}$
in \eqref{eq:updatebeta}, we use the standard adaptive method termed \texttt{ADAM}, by \cite{kingma2014adam}, known as stabilising the unnecessary numerical instability
due to the choice of $\{\kappa_{t}\}$. Assume that after $t$ steps we have $c_{t}:=\nabla\hat{\mathcal{H}}_{t}(\gamma_t)$.
\texttt{ADAM} applies the following iterative procedure,
\begin{gather*}
m_{t}=m_{t-1}\beta_{1}+(1-\beta_{1})c_{t}, \quad  v_{t}=v_{t-1}\beta_{2}+(1-\beta_{2})c_{t}^{2},\\
\hat{m}_{t}=m_{t}/(1-\beta_{1}^{t}),\quad \hat{v}_{t}=v_{t}/(1-\beta_{2}^{t}),\\
\gamma_{t+1}=\gamma_{t}-\alpha\hat{m}_{t}/(\sqrt{\hat{v}_{t}}+\epsilon),
\end{gather*}
where $(\beta_{1},\beta_{2},\alpha,\epsilon)$ are the tuning parameters.
The convergence properties of \texttt{ADAM} have been
widely studied \citep{kingma2014adam,reddi2019convergence}. Following closely \citet{kingma2014adam}, in all uses of \texttt{ADAM} below we set $(\beta_{1},\beta_{2},\alpha,\epsilon)=(0.9,0.999,0.003,10^{-8})$.
\texttt{ADAM} is nowadays a standard and very effective addition to the type of recursive inference algorithms we are considering here, even more so as for increasing dimension of unknown parameters. See the above references for more motivation and details.

\subsection{Simulation Studies}

We here demonstrate the numerical performance of the proposed method. 
Throughout this study, we consider Gaussian models $\{f_{\theta}:\theta\in\Theta\}=\text{\ensuremath{\mathcal{N}(\mu,\sigma^{2})}}$ where both $\mu$ and $\sigma^2$ are unknown parameters.
We first simulated data $\{y_{i}\}_{i=1}^{100}\overset{i.i.d.}{\sim}\mathcal{N}(1,1)$ and then randomly replaced $\tau\%$ of $\{y_{i}\}_{i=1}^{100}$ by $y_{i}+5$. 
This setting is commonly referred to as the M-open world \citep{bernardo2009bayesian} in the sense that there is no $\theta^{\star}$ such that $g=f_{\theta^{\star}}$.

\subsection*{Experiment 1: Convergence property}
We investigate convergence behaviour of \ref{alg:SMC-sampler-.}.
In this study, we set $(N,T,\gamma_{0})=(2,000,500,0.1)$ with $50$ MCMC steps in \ref{alg:SMC-sampler-.} to estimate $\gamma$.
The results are shown in \ref{fig:Example3}. As can be seen in the figure, our proposed method converges stably to the true value after about 100 iterations.
Here, the true value was obtained by first approximating the H-score with MCMC in the same way as before and then using a grid search to find the $\gamma$ that minimises it.

\begin{figure}[H]
\begin{centering}
\includegraphics[scale=0.4]{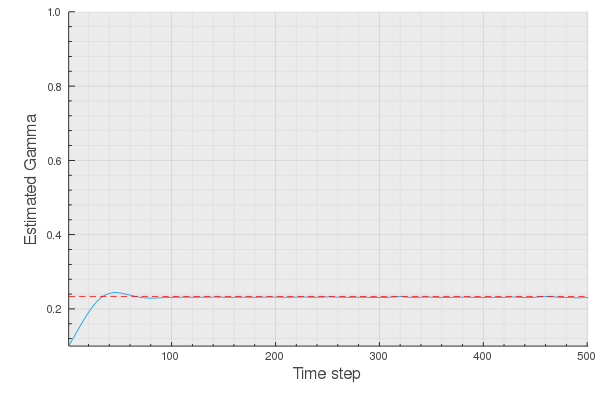}\includegraphics[scale=0.4]{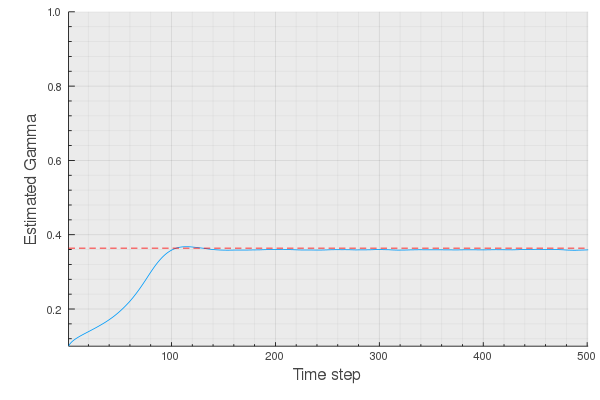}
\par\end{centering}
\caption{\label{fig:example4.1}Gaussian models experiment: Trajectories from
execution of \ref{alg:SMC-sampler-.}. We used $N=2,000$ particles
with $50$ MCMC iterations and initial value $\gamma_{0}=0.1$. The
left panel shows results when $\tau=5$ and the right one shows when
$\tau=10$. The horizontal dashed lines in the plots show the true
parameter $\gamma^{\star}=0.2339$ for the left and $\gamma^{\star}=0.3638$
for the right. The blue lines show the trajectory of $\hat{\gamma}$ estimated by \ref{alg:SMC-sampler-.}.}
\label{fig:Example3}
\end{figure}

\subsection*{Experiment 2: Comparison with methods using fixed values of $\gamma$}
We next compare the performance of our proposed method with that of a non-adaptive method using a fixed value of $\gamma$.
We set $\tau\in \{0, 10,20, 30\}$, and for each case we computed the posterior distribution of $\mu$ using \ref{alg:SMC-sampler-.} and the vanilla version of MCMC.
For \ref{alg:SMC-sampler-.}, the tuning parameters were set to $(N,T,\gamma_{0})=(2,000,300,0.1)$ with $50$ MCMC steps, and for vanilla MCMC, $\gamma$ was set to $\gamma\in \{0.1,0.3,0.5,0.7,0.9\}$ in advance of the estimation with 100,000 MCMC steps.
We used Metropolis-Hastings kernels based on a Gaussian random walk with $\mathcal{N}(0,0.4)$ for the two cases under the uniform prior.
Using the posterior samples obtained from \ref{alg:SMC-sampler-.} and non-adaptive methods, we computed the posterior mean and $95\%$ credible interval of $\mu$. 
We ran 100 Monte Carlo experiments to calculate their (empirical) mean square error (MSE) and average 95$\%$ credible interval (ACI). 
The MSE is computed against the target value of 1, and the value is multiplied by 100. 
The results are given in \ref{table:nonadaptive}.
Although it is a simple example, the results summarised in the table clearly show that the accuracy of the inference is improved by estimating $\gamma$ from the data rather than simply fixing it in terms of MSE.
In fact, the best $\gamma$ among the five choices depends on the underlying contamination ratio that we do not know in practice. 
Hence, it is difficult to determine a suitable value of $\gamma$ simply by looking at the data, while our method can automatically tune the value of $\gamma$ from the data.
It should also be noted that the importance of adaptive tuning of $\gamma$ is reflected in the results of not only MSE but also ACI; that is, the interval length obtained from \ref{alg:SMC-sampler-.} is narrow compared with the non-adaptive methods in all the four scenarios.

\begin{table}[H]
\begin{centering}
\begin{tabular}{|ccccc|}
\hline 
 & $\tau=0$ & $\tau=10$ & $\tau=20$ & $\tau=30$\tabularnewline
\hline 
$(\gamma=0.1)$ &  &  &  & \tabularnewline
MSE & {\bf 3.66} & 12.09 & 7.58 & 21.54\tabularnewline
ACI & (0.62, 1.06) & (1.08, 1.57) & (1.00, 1.49) & (1.18, 1.71)\tabularnewline
\hline 
$(\gamma=0.3)$ &  &  &  & \tabularnewline
MSE & 4.87 & {\bf 4.08} & {\bf 2.27} & {\bf 2.71}\tabularnewline
ACI & (0.56, 1.10) & (0.82, 1.44) & (0.69, 1.28) & (0.65, 1.27)\tabularnewline
\hline 
($\gamma=0.5$) &  &  &  & \tabularnewline
MSE & 6.27 & 4.79 & 3.32 & \textcolor{black}{4.86}\tabularnewline
ACI & (0.49, 1.13) & (0.69, 1.45) & (0.61, 1.31) & (0.52, 1.24)\tabularnewline
\hline 
($\gamma=0.7$) &  &  &  & \tabularnewline
MSE & 8.00 & 5.78 & 4.61 & 6.78\tabularnewline
ACI & (0.41, 1.18) & (0.58, 1.51) & (0.55, 1.38) & (0.43, 1.29)\tabularnewline
\hline 
($\gamma=0.9$) &  &  &  & \tabularnewline
MSE & 10.04 & 8.35 & 6.44 & 8.90\tabularnewline
ACI & (0.33, 1.24) & (0.46, 1.59) & (0.47, 1.46) & (0.34, 1.35)\tabularnewline
\hline 
\hline
$\hat{\gamma}$ & 0.006 & 0.207 & 0.213 & 0.272\tabularnewline
MSE & 3.13 & 3.51 & 2.13 & 2.47\tabularnewline
ACI & (0.66, 1.05) & (0.91, 1.46) & (0.76, 1.32) & (0.67, 1.28)\tabularnewline
\hline 
\end{tabular}
\par\end{centering}
\centering{}
\caption{Empirical mean squared errors (MSE) and average 95$\%$ credible intervals (ACI) of \ref{alg:SMC-sampler-.} and the non-adaptive method (the vanilla version of MCMC with fixed $\gamma$), based on 100 Monte Carlo experiments. 
The best MSE value among different choices of $\gamma$ is highlighted in bold. 
The bottom row shows estimated $\gamma$ and the corresponding MSE and CI when estimated with our proposed method. The tuning parameters were set to $(N,T,\gamma_{0})=(2,000,300,0.1)$ with $50$ MCMC steps
}
\label{table:nonadaptive}
\end{table}

\subsection*{Experiment 3: Comparison with \cite{jewson2021general}}
We next compare the proposed method with the H-posterior proposed by \cite{jewson2021general} (denoted by JR hereafter), where the posterior of the model parameters $(\mu,\sigma^2)$ as well as $\gamma$ can be obtained.
To apply the JR method, we generated 1000 posterior samples after discarding the first 500 samples. 
We evaluate the performance of the inference of $\mu$ and $\sigma$ by MSE (multiplied by 100), coverage probability (CP) and average length (AL) of $95\%$ credible intervals.
The results are shown in Table~\ref{table:JR}, where the average estimates of $\gamma$ are also shown. 
Although both methods provide similar estimates of $\gamma$, the accuracy of point estimation of JR is slightly better than that of \ref{alg:SMC-sampler-.}.
However, it is observed that JR tends to produce a short coverage length so that the CP of the JR method is much smaller than the nominal level $95\%$. 
Accordingly, the average length is much smaller than those by \ref{alg:SMC-sampler-.}.
This means that a direct application of the H-posterior by \cite{jewson2021general} may fail to capture the uncertainty of the posterior compared with the proposed method.

\begin{table}[H]
\begin{centering}
\begin{tabular}{|ccccccccccc|}
\hline 
&&&&& \multicolumn{2}{c}{MSE} & \multicolumn{2}{c}{CP} & \multicolumn{2}{c}{AL} \\
 & $\tau$ &  & mean$(\hat{\gamma})$ & var$(\hat{\gamma})$ & $\mu$ & $\sigma$ & $\mu$ & $\sigma$ & $\mu$ & $\sigma$ \\
 \hline
 & 5 &  & 0.194 & 0.001 & 3.62 & 5.23 & 97 & 84 & 0.57 & 0.53 \\
\ref{alg:SMC-sampler-.} & 10 &  & 0.297 & 0.003 & 4.57 & 9.20 & 97 & 88 & 0.51 & 0.56 \\
 & 15 &  & 0.377 & 0.009 & 1.06 & 12.60 & 92 & 82 & 0.61 & 0.67 \\
 \hline
 & 5 &  & 0.193 & 0.011 & 2.33 & 5.10 & 68 & 13 & 0.27 & 0.16 \\
JR & 10 &  & 0.230 & 0.009 & 2.46 & 6.25 & 64 & 10 & 0.28 & 0.16 \\
 & 15 &  & 0.261 & 0.013 & 2.50 & 7.27 & 68 & 5 & 0.29 & 0.15 \\
\hline 
\end{tabular}
\par\end{centering}
\centering{}
\caption{Mean squared errors (MSE), coverage probability (CP) and average length (AL) of $95\%$ credible intervals of $\mu$ and $\sigma$ based on \ref{alg:SMC-sampler-.} and the JR method. MSE is multiplied by 100.
The tuning parameters in our algorithm were set to $(N,T,\gamma_{0})=(2,000,300,0.1)$ with $5$ MCMC steps.
}
\label{table:JR}
\end{table}

\subsection*{Experiment 4: Comparison with tempering}
Following \cite{nakagawa2020robust}, we compare the robustness to outliers for the two generalised posterior distributions. 
The first distribution is constructed in the same way as before, while the other is constructed using tempering.
We specify a tempered posterior as $\ensuremath{\Pi}_{\phi_{t}}(\theta\mid y_{1:n})\propto\mathcal{L}(y_{1:n};\theta)^{\phi_t}\pi(\theta)$ where $0=\phi_0<\phi_{1}\cdots<\cdots\phi_{T}=1$. 
To construct the sequence, we divided the interval $[0,1]$ into 500 equal parts.  
We applied \ref{alg:SMC-sampler-.} and the SMC sampler with the tempered posterior to test the robustness of the proposed method to data sets containing outliers.
We used $N=2,000$ particles, $50$ MCMC steps for both methods, and set $(T,\gamma_0)=(500,0.1)$ for \ref{alg:SMC-sampler-.}.
The prior and MCMC kernels were set as the previous experiment, and the density estimation results obtained from the estimation results are summarised in \ref{fig:posresult}. 
The red line represents the posterior density estimate for $\mu$ when the data do not contain any outliers, and the blue line represents it when the data contain outliers.
It is clear from the estimation results that our proposed method is robust even when the dataset contains outliers, while the SMC sampler with the tempered posterior is greatly affected by outliers, and the estimated posterior distributions are completely separated as a result.

\begin{figure}[H]
\begin{centering}
\includegraphics[scale=0.6]{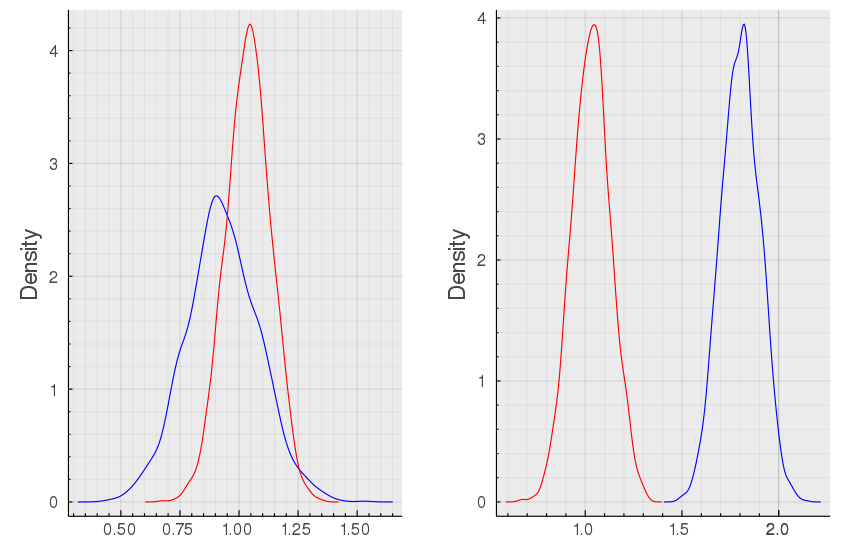}
\par\end{centering}
\centering{}\caption{The density estimation of $\mu$ estimated by \ref{alg:SMC-sampler-.} (left) and SMC sampler with the tempered posterior (right). The blue line shows when $\tau=20$ (contains $20\%$ outliers) and the red one shows when $\tau=0$  (contains $0\%$ outliers) in both panels.}
\label{fig:posresult}
\end{figure}

\subsection{Applications to Real Data}
\subsection*{Newcomb data}
We apply our methodology to Simon Newcomb's measurements of the speed of light data, motivated by applications in \citet{stigler1977robust,basu1998robust,basak2021optimal}.
The data can be obtained from Andrew Gelman's webpage: \url{http://www.stat.columbia.edu/~gelman/book/data/light.asc}.
The sample size of the data set is 66 and contains two outliers,
-44 and -2, illustrated in \ref{fig:data}. We fitted a Gaussian distribution
model $\{f_{\theta}:\theta\in\Theta\}=\text{\ensuremath{\mathcal{N}(\mu,\sigma^{2})}}$
to the data and used \ref{alg:SMC-sampler-.} to obtain the posterior
distribution of the parameters $(\mu,\sigma)$. The tuning parameters $(N,T,\gamma_{0})$ in \ref{alg:SMC-sampler-.}
were set to $(2,000,300,0.1)$ with $50$ MCMC iterations.
The MCMC kernel was constructed as in the previous examples, and results are given in \ref{fig:Estimated-pos}.
The existing study \citep{basak2021optimal} reported $\hat{\gamma}=0.23$ for the same data set, which is very high compared to our estimate result of $\hat{\gamma}=0.0855$.
Since the method proposed in \cite{basak2021optimal} requires a pilot plot and the estimation results depend significantly on it, we believe our estimation results are more reasonable. In fact, it is unlikely that we will have to use a value of ${\gamma}=0.23$ for a data set that contains only two outliers. As shown in \cite{basu1998robust}, the parameter estimates are almost the same when $\gamma=0.0855$ and when $\gamma=0.23$. However, from the point of view of statistical efficiency, it would be preferable to adopt the lower value of $\gamma=0.0855$ if the estimates were the same.
To confirm this, $100$ bootstrap resamplings were performed on the data, and the posterior bootstrap mean of each parameter and the variance was calculated, reported in \ref{table:newcom}.
For each re-sampled data, we compared the results when the posterior distribution was calculated while estimating $\gamma$ with our method and when the posterior distribution was calculated using MCMC after estimating and fixing it with the method proposed in \citet{basak2021optimal}.
The numerical experiments show that although the means estimated parameters agree between the two methods, the variances are much smaller for our method, suggesting that overestimation of $\gamma$ leads to statistical inefficiency.

\begin{figure}[H]
\begin{centering}
\includegraphics[scale=0.6]{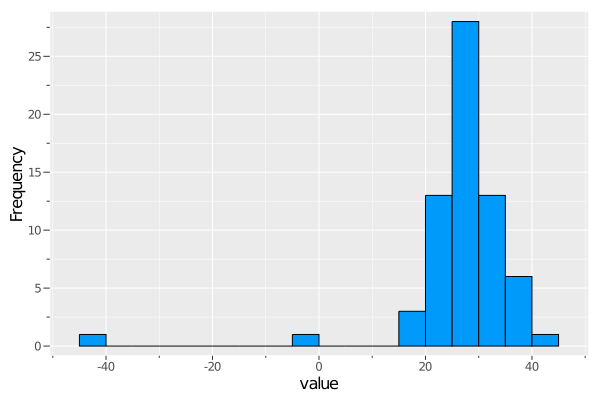}
\par\end{centering}
\centering{}\caption{The histogram of Simon Newcomb's measurements of the speed of light
data.}
\label{fig:data}
\end{figure}

\begin{figure}[H]
\begin{centering}
\includegraphics[scale=0.6]{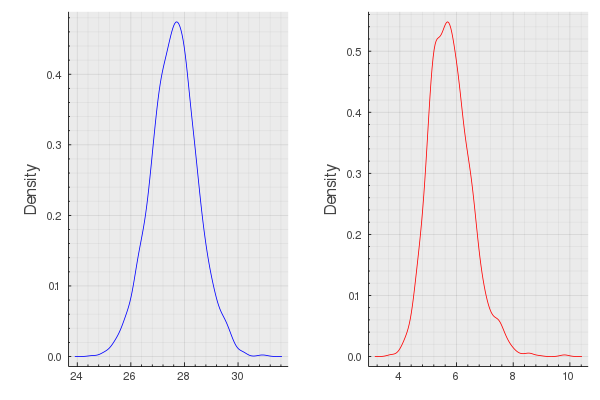}
\caption{The density estimation of $\mu$ (left) and
$\sigma$ (right) for Simon Newcomb's measurements of the speed of light
data. The tuning parameters $(N,T,\gamma_{0})$ were set to $(2,000,300,0.1)$ 
with $50$ MCMC iterations. The mean value of the estimated $\mu$ is 27.6082 and $\sigma$ is 5.7829 with $\hat{\gamma}=0.0855$.}
\label{fig:Estimated-pos}
\par\end{centering}
\end{figure}

\begin{table}[H]
\begin{centering}
\begin{tabular}{|ccccc|}
\hline 
& mean($\hat{\mu}$) & var($\hat{\mu}$) & mean($\hat{\sigma}$) & var($\hat{\sigma}$)\tabularnewline
\hline 
\ref{alg:SMC-sampler-.} & 27.559 & 0.0017 & 5.7605 & 0.0013\tabularnewline
\hline 
\citet{basak2021optimal}  & 27.674 & 0.4351 & 5.3976 & 0.5295\tabularnewline
\hline
\end{tabular}
\par\end{centering}
\centering{}
\caption{Mean and variance of the posterior means of the parameters from $100$ bootstrap re-samplings.
The first row shows the result when using the proposed method, and the second one shows when using the method studied in \citet{basak2021optimal}
}
\label{table:newcom}
\end{table}

\subsection*{Hertzsprung–Russell Star Cluster Data}
We use our methodology to perform linear regression models with normal errors, that is $y_i \sim \mathcal{N}(x^{\top}_{i}\beta,\sigma^2)$. Motivated by \cite{basak2021optimal}, we fitted the regression model to the Hertzsprung-Russell star cluster data \citep{rousseeuw2005robust}, without constants. 
The data set contains 47 observations on the logarithm of the effective temperature at the surface of the CYG OB1 star cluster (Te, covariates $\{x_i\}$) and the logarithm of its light intensity (L/L0, explained variables $\{y_i\}$). 
The data can be obtained from \url{https://rdrr.io/cran/robustbase/man/starsCYG.html}, and shown at \ref{fig:Hertzsprung}. 
The tuning parameters $(N,T,\gamma_{0})$ in \ref{alg:SMC-sampler-.} were set to $(2,000,300,0.1)$ with $50$ MCMC iterations, and we used the uniform prior for $(\beta,\sigma)$. 
The MCMC kernel was constructed as in the previous examples, and results are given in \ref{fig:Estimated-pos-reg}. 
The corresponding OLS estimates were $(\hat{\beta},\hat{\sigma})=(1.1559,0.7219)$. Whilst we obtained $\hat{\gamma }=0.1165$, \cite{basak2021optimal} reported $\hat{\gamma}=0.76$ for the same data set.
Our numerical experiments and previous studies \citep{ghosh2016robust,nakagawa2020robust} will suggest that as the proportion of outliers in the data increases, the value of $\gamma$ also tends to increase.
Thus, such a large value of $\gamma$ is not reasonable considering the proportion of outliers in the data (only four samples in the lower right part in \ref{fig:Hertzsprung}), suggesting the superiority of our proposed method.
Indeed, to confirm the suggested statistical inefficiency, the same experiments as in the previous section were carried out, and the results are summarised in \ref{table:star}.
Although the results are not as striking as in the previous Newcomb data example, it would be possible to confirm that statistical inefficiencies occur in the estimation by the proposed method in \cite{basak2021optimal} in the regression model as well.

\begin{figure}[H]
\begin{centering}
\includegraphics[scale=0.6]{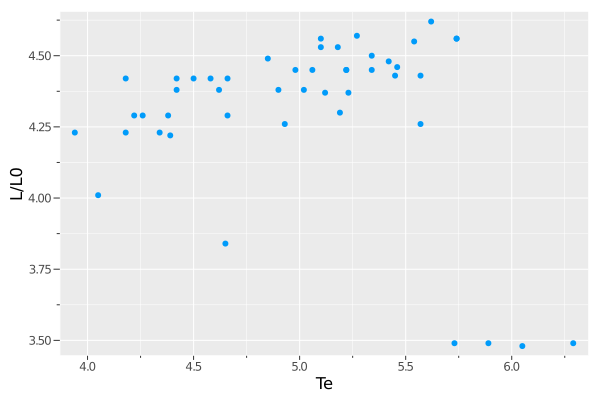}
\centering{}\caption{The scatter plot of Hertzsprung–Russell star cluster data.}
\label{fig:Hertzsprung}
\par\end{centering}
\end{figure}

\begin{figure}[H]
\begin{centering}
\includegraphics[scale=0.6]{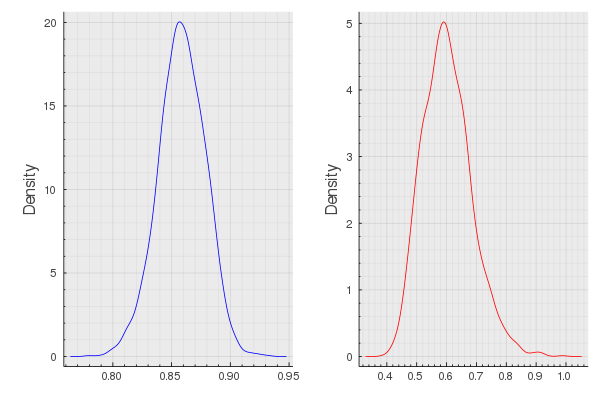}
\caption{The density estimation of $\beta$ (left) and
$\sigma$ (right) for Hertzsprung–Russell star cluster data. The tuning parameters $(N,T,\gamma_{0})$ were set to $(2,000,300,0.1)$ 
with $50$ MCMC iterations. The mean value of the estimated $\beta$ is 0.8586 and $\sigma$ is 0.602 with $\hat{\gamma}=0.1165$. The corresponding OLS estimates are $(\hat{\beta},\hat{\sigma})=(1.1559,0.7219)$}
\label{fig:Estimated-pos-reg}
\par\end{centering}
\end{figure}

\begin{table}[H]
\begin{centering}
\begin{tabular}{|ccccc|}
\hline 
& mean($\hat{\mu}$) & var($\hat{\mu}$) & mean($\hat{\sigma}$) & var($\hat{\sigma}$)\tabularnewline
\hline 
\ref{alg:SMC-sampler-.} & 0.8514 & 0.0002 & 0.6100 & 0.0037\tabularnewline
\hline 
\citet{basak2021optimal}  &0.8763 & 0.0299 & 0.6536 & 0.0708\tabularnewline
\hline
\end{tabular}
\par\end{centering}
\centering{}
\caption{Mean and variance of the posterior means of the parameters from $100$ bootstrap re-samplings.
The first row shows the result when using the proposed method, and the second one shows when using the method studied in \citet{basak2021optimal}
}
\label{table:star}
\end{table}

\section{Concluding remarks} \label{sec:coclusion}
Our proposed method performs reasonably well and provides one of the few options, as far as we know, for routine robust Bayesian inference. To the best of our knowledge, this is the first attempt to propose both a theory and a computational algorithm to estimate the tuning parameters from data and to allow robust Bayesian estimation. We have shown numerically that a more efficient Bayesian estimation can be achieved by estimating the tuning parameter $\gamma$ from the data. We have proposed an efficient sampling method using SMC samplers considering the sequence of $\gamma$ as the temperature. Compared to existing studies \citep{warwick2005choosing,basak2021optimal}, our method has specificity and usefulness in that we can estimate the tuning parameters and sample from the posterior distribution simultaneously, and pilot plots and the asymptotic variance formula are not necessary. In this paper, we have focused in particular on the case of the density power divergence, but we want to stress that our method is general enough in the sense that it can be applied to the Bayesian estimation of other robust divergence-induced models.

Furthermore, our framework opens up a number of routes for future research and insight, including those described below.
\begin{enumerate}
    \item As we have noted, the integral term $\int f_{\theta}(x)^{1+\gamma}dx$ is eliminated in the H-score, while the computation of the posterior distribution is computationally expensive, so it may be better to consider the H-posterior studied in \cite{jewson2021general} in this respect. However, the robustness of the H-posterior has not yet been studied, and it would therefore be interesting in the future to investigate this point in more detail using the influence function.
    \item Another direction of investigation involves the construction of an efficient MCMC kernel for posterior distributions derived from robust divergence such as (\ref{eq:posterior}). To make good inferences from data containing outliers, the posterior distribution induced by robust divergence is a model with a more or less heavy-tailed.
    As studied in \cite{kamatani2018efficient}, many standard MCMC algorithms are known to perform poorly in such cases, especially in higher dimensions. Therefore, studying MCMC algorithms within \ref{alg:SMC-sampler-.} tailored to the posterior distribution induced by robust divergence would allow for more efficient Bayesian robust estimation.
    \item In this study, we have not focused on time series data, but, as \cite{shao2019bayesian} shows, the H-score can also be defined for models that are not independent, for example, state-space models. In fact, \cite{boustati2020generalized} proposes a method for Bayesian filtering of state-space models using robust divergence, but the tuning parameters need to be estimated before filtering, and in this sense, it is not online filtering. \ref{alg:SMC-sampler-.} does batch estimation, but we believe that extending it to online estimation would allow robust filtering of the state-space model while estimating the tuning parameters online from the data.
    \item This study has focused on estimation and computational methods proposed in the generalised Bayesian framework, particularly using robust divergence. 
    On the other hand, methods using the Maximum Mean Discrepancy \citep{cherief2020mmd} and Kernel Stein Discrepancy \citep{matsubara2021robust} have also been proposed in recent years in the same generalised Bayesian framework, although not with the motivation of dealing with outliers.
    Both require adjustment of the hyperparameters of the kernel used, and it may be possible to estimate them using our proposed method and compare their performance.
    We have avoided comparing these potential alternative approaches because we believe this would obscure the main messages we have tried to convey within the numerical results section. Such a detailed numerical study can be the subject of future work.
\end{enumerate}

\section*{Acknowledgement}
SY was supported by the Japan Society for the Promotion of Science (KAKENHI) under grant number 21K17713. SS was supported by the Japan Society for the Promotion of Science (KAKENHI) under grant number 21H00699.

\bibliographystyle{apalike}
\bibliography{reference}

\appendix

\section{Proof of Proposition \ref{prop:monotones}} \label{sec:proofmonotone}
Recall that $\log\mathcal{L}^{\mathbf{R}}_{\gamma}(y_{i};\theta)=\frac{1}{\gamma}f_{\theta}(y_{i})^{\gamma}-\frac{1}{1+\gamma}\int f_{\theta}(x)^{1+\gamma}dx-\frac{1}{\gamma}+1$. 
First, $\frac{1}{\gamma}f_{\theta}(y_{i})^{\gamma}-\frac{1}{\gamma}$ is an increasing function of $\gamma$ for arbitrary $f_{\theta}(y_{i})$ and $\gamma>0$.
Furthermore, since we have assumed that $f_{\theta}(y_{i})\leq 1$, it holds that $f_{\theta}(x)^{1+\gamma_1}\geq f_{\theta}(x)^{1+\gamma_2}$ for $\gamma_1<\gamma_2$ and arbitrary $x$, so that $\frac{1}{1+\gamma}\int f_{\theta}(x)^{1+\gamma}dx$ is a decreasing function of $\gamma$.
Hence, $\log\mathcal{L}^{\mathbf{R}}_{\gamma}(y_{i};\theta)$ is increasing since it is a sum of two increasing functions.

\section{Derivation of \eqref{eq:H}} \label{sec:derivation}
For simplicity, we consider $d_y=1$, but the extension to the general dimension is straightforward. 
Let $p(y|\theta)$ be a general model with parameter $\theta$ and $p(y)=\int p(y|\theta)\Pi(\theta)d\theta$ be the marginal likelihood under prior $\Pi(\theta)$. 
Under the assumptions stated in supplement (S6) in \citet{shao2019bayesian}, the following identify always holds
\begin{align*}
&\sum_{k=1}^{d_y}\left\{2\frac{\partial^2\log p(y)}{\partial y_{(k)}^2}+\left( \frac{\partial \log p(y)}{\partial y_{(k)}}\right)^2\right\} \\
& \ \ \ \ \ \  
=\sum_{k=1}^{d_y}\left\{\mathbb{E}\left[2\frac{\partial^2\log p(y\mid\theta)}{\partial y_{(k)}^2}+ \left( \frac{\partial \log p(y\mid\theta)}{\partial y_{(k)}}\right)^2\mid y \right]
-\left(\mathbb{E}\left[\frac{\partial \log p(y\mid \theta)}{\partial y_{(k)}}\mid y \right] \right)^2\right\},
\end{align*}
where the expectation is taken with respect to the posterior distribution of $\theta$ given $y$.
Using the above identity with $p(y)=p(y_i|y_{-i})$, $p(y|\theta)=p(y_i|\theta,y_{-i})$ and $\Pi(\theta)=\Pi(\theta|y_{-i})$, where $y_{-i}=(y_1,...,y_{i-1},y_{i+1},...,y_n)$, we have 
\begin{align*}
&\sum_{k=1}^{d_y}\left\{2\frac{\partial^2\log p(y_i\mid y_{-i})}{\partial y_{i(k)}^2}+\left( \frac{\partial \log p(y_i\mid y_{-i})}{\partial y_{i(k)}}\right)^2\right\}\\
& \ \ \ \ \ 
=\sum_{k=1}^{d_y}\bigg\{\mathbb{E}\left[2\frac{\partial^2\log p(y_i\mid y_{i-1},\theta)}{\partial y_{i(k)}^2}
+\left( \frac{\partial \log p(y_i\mid y_{i-1},\theta)}{\partial y_{i(k)}}\right)^2\mid y_i,y_{i-1} \right]\\
& \ \ \ \ \ \ \ \ \ 
-\left(\mathbb{E}\left[\frac{\partial \log p(y_i\mid y_{i-1},\theta)}{\partial y_{i(k)}}\mid y_i,y_{i-1} \right] \right)^2\bigg\}.
\end{align*}
Notice that, since we have assumed $i.i.d.$ observations, we have $p(y_i\mid y_{i-1},\theta)=p(y_i\mid \theta)$.
Hence, the expression (\ref{eq:H}) follows by setting $p(y_i\mid \theta)=\mathcal{L}_{\gamma}(y_{i};\theta)$.

\section{Proof of Proposition \ref{pro:consistency}} \label{sec:proofconsistency}
The proof here is essentially the same as \citet{shao2019bayesian} and \citet{jewson2021general}, so we only provide an overview of the proof.
Assume that, for simplicity, $d_y=1$.
First one can show that $\mathcal{H}_n(\gamma)$ can be decomposed into the sum of conditional expectation terms of $\mathcal{H}(y_i,p_{\gamma}(y_i\mid\theta,y_{-i}))=\mathcal{H}(y_i,p_{\gamma}(y_i\mid\theta))$, and the sum of conditional variance terms of $\frac{\partial\log p_\gamma(y_i\mid\theta)}{\partial y_i}$, see \citet{dawid2015bayesian,shao2019bayesian}. 
Then under the assumptions stated in supplement of \cite{shao2019bayesian}, the variance term will converge at $0$ w.p.1.
Let $(\mathbf{B},\|\cdot\|)$ be the space of continuous real functions on the compact set of $\gamma$ equipped with the sup-norm.
Then, under the same assumptions, $\frac{1}{n}\sum_{i=1}^n\mathbb{E}_{\gamma}[\mathcal{H}(y_i,p_{\gamma}(y_i\mid\theta))]$ may take values in this space.
Then strong law of large numbers on a separable Banach space \citep{azlarov1982laws,beskos2009monte}, may be applied to $\frac{1}{n}\sum_{i=1}^n\mathbb{E}_{\gamma}[\mathcal{H}(y_i,p_{\gamma}(y_i|\theta))]-\mathbb{E}_g[\mathcal{H}(y_1,p_{\gamma}(y_1\mid\theta))]$.
Combining this result with (s10) in supplement of \cite{shao2019bayesian} and integration by parts \citep{hyvarinen2005estimation,dawid2015bayesian} would yield $\lim_n\sup_{\gamma}\frac{1}{n}\mathcal{H}_n(\gamma)=\mathcal{J}(\gamma)$ w.p.1.
The result follows under the assumption for identification such that $\mathcal{J}(\gamma)$ is only maximised at $\gamma^{\star}$, see A1. in \cite{jewson2021general} for instance.

\section{Derivatives of the H-score} \label{sec:derivative_Hscore}
In the following argument, we assume the exchangeability of integral and derivative without any remarks.
For simplicity, we consider a univariate case, namely $d_y=1$.
We define $\mathcal{D}_{\gamma}(y_{1:n};\theta):=\log\mathcal{L}_{\gamma}(y_{1:n};\theta)$ and $\mathcal{D}_{\gamma}(y_i;\theta):=\log\mathcal{L}_{\gamma}(y_i;\theta)$.  
The derivative of the H-score with respect to $\gamma$ is expressed as 
\begin{align*}
\frac{d}{d\gamma}\mathcal{H}_{n}(\gamma)
&=2\sum_{i=1}^n\int\frac{d}{d\gamma}\left[\left\{\frac{\partial^{2}\mathcal{D}_{\gamma}(y_i;\theta)}{\partial^{2}y_i}+\left(\frac{\partial \mathcal{D}_{\gamma}(y_i;\theta)}{\partial y_i}\right)^{2}\right\}\Pi_{\gamma}(\theta\mid y_{1:n})\right]d\theta\\
&-2\sum_{i=1}^n\mathbb{E}\left[\frac{\partial \mathcal{D}_{\gamma}(y_i;\theta)}{\partial y_i}\mid y_{1:n}\right]
\times \int\frac{d}{d\gamma}\left\{\frac{\partial \mathcal{D}_{\gamma}(y_i;\theta)}{\partial y_i}\Pi_{\gamma}(\theta\mid y_{1:n})\right\}d\theta,
\end{align*}
which requires the computation of integral of the following form:
\begin{equation}\label{eq:H-deriv}
\int \frac{d}{d\gamma}\Big\{C_{\gamma}^{(k)}(y_i;\theta)\Pi_{\gamma}(\theta\mid y_{1:n})\Big\}d\theta
=
\int \frac{d}{d\gamma}\left\{C_{\gamma}^{(k)}(y_i;\theta)\frac{e^{\mathcal{D}_{\gamma}(y_{1:n};\theta)}\pi(\theta)}{\int e^{\mathcal{D}_{\gamma}(y_{1:n};\theta)}\pi(\theta)d\theta}\right\}d\theta,
\end{equation}
where 
\begin{equation}\label{C}
C_{\gamma}^{(1)}(y_i;\theta)=\frac{\partial^{2}\mathcal{D}_{\gamma}(y_i;\theta)}{\partial^{2}y_i}+\left(\frac{\partial \mathcal{D}_{\gamma}(y_i;\theta)}{\partial y_i}\right)^{2}, \ \ \ \ 
C_{\gamma}^{(2)}(y_i;\theta)=\frac{\partial \mathcal{D}_{\gamma}(y_i;\theta)}{\partial y_i},
\end{equation}
and $\mathcal{D}_{\gamma}(y_{1:n};\theta):=\log\mathcal{L}_{\gamma}(y_{1:n};\theta)$.
It follows that 
\begin{align*}
& \ \ \ 
\int \frac{d}{d\gamma}\left\{C_{\gamma}^{(k)}(y_i;\theta)\frac{e^{\mathcal{D}_{\gamma}(y_{1:n};\theta)}\pi(\theta)}{\int e^{\mathcal{D}_{\gamma}(y_{1:n};\theta)}\pi(\theta)d\theta}\right\}d\theta\\
&=
\int \left\{\frac{d}{d\gamma}C_{\gamma}^{(k)}(y_i;\theta)\right\}\Pi_{\gamma}(\theta\mid y_{1:n})d\theta
+
\int C_{\gamma}^{(k)}(y_i;\theta)\left\{\frac{d}{d\gamma}\mathcal{D}_{\gamma}(y_{1:n};\theta)\right\}\Pi_{\gamma}(\theta\mid y_{1:n})d\theta\\
& \ \ - \int C_{\gamma}^{(k)}(y_i;\theta)\frac{e^{\mathcal{D}_{\gamma}(y_{1:n};\theta)}\pi(\theta)}{\left\{\int e^{\mathcal{D}_{\gamma}(y_{1:n};\theta)}\pi(\theta)d\theta\right\}^2}\times \int \pi(\theta)e^{\mathcal{D}_{\gamma}(y_{1:n};\theta)}\left\{\frac{d}{d\gamma}\mathcal{D}_{\gamma}(y_{1:n};\theta)\right\}d\theta,
\end{align*}
where the third term is further simplified to 
\begin{align*}
\int C_{\gamma}^{(k)}(y_i;\theta)\Pi_{\gamma}(\theta\mid y_{1:n}) d\theta
\times
\int \left\{\frac{d}{d\gamma}\mathcal{D}_{\gamma}(y_{1:n};\theta)\right\}\Pi_{\gamma}(\theta\mid y_{1:n}) d\theta.
\end{align*}
Hence, the derivative (\ref{eq:H-deriv}) is expressed as 
\begin{equation}
\begin{split}
&\mathbb{E}\left[ \left\{\frac{d}{d\gamma}C_{\gamma}^{(k)}(y_i;\theta)\right\} 
+
C_{\gamma}^{(k)}(y_i;\theta)\left\{\frac{d}{d\gamma}\mathcal{D}_{\gamma}(y_{1:n};\theta)\right\}
\mid y_{1:n}\right]\\
&  \ \ \ \ 
-\mathbb{E}\left[ C_{\gamma}^{(k)}(y_i;\theta) \mid y_{1:n}\right]\times
\mathbb{E}\left[\frac{d}{d\gamma}\mathcal{D}_{\gamma}(y_{1:n};\theta) \mid y_{1:n}\right].
\end{split}
\end{equation}
Finally, the derivative of the H-score is expressed as 
\begin{equation}\label{derive-H-score}
\begin{split}
\frac{d}{d\gamma}\mathcal{H}_{n}(\gamma)&=2\sum_{i=1}^n \ \mathbb{E}\left[ \left\{\frac{d}{d\gamma}C_{\gamma}^{(1)}(y_i;\theta)\right\} 
+
C_{\gamma}^{(1)}(y_i;\theta)\left\{\frac{d}{d\gamma}\mathcal{D}_{\gamma}(y_{1:n};\theta)\right\}
\mid y_{1:n}\right]\\
&  \ \ \ \ 
- 2\sum_{i=1}^n \mathbb{E}\left[ C_{\gamma}^{(1)}(y_i;\theta) \mid y_{1:n}\right]  
\mathbb{E}\left[\frac{d}{d\gamma}\mathcal{D}_{\gamma}(y_{1:n};\theta) \mid y_{1:n}\right]\\
& -2\sum_{i=1}^n \mathbb{E}\left[\frac{d}{dy_i}\mathcal{D}_{\gamma}(y_i;\theta) \mid y_{1:n}\right]\mathbb{E}\left[ \left\{\frac{d}{d\gamma}C_{\gamma}^{(2)}(y_i;\theta)\right\}+C_{\gamma}^{(2)}(y_i;\theta)\left\{\frac{d}{d\gamma}\mathcal{D}_{\gamma}(y_{1:n};\theta)\right\}
\mid y_{1:n}\right]\\
&  \ \ \ \ 
+ 2\sum_{i=1}^n \mathbb{E}\left[\frac{d}{dy_i}\mathcal{D}_{\gamma}(y_i;\theta) \mid y_{1:n}\right] \mathbb{E}\left[ C_{\gamma}^{(2)}(y_i;\theta) \mid y_{1:n}\right]  
\mathbb{E}\left[\frac{d}{d\gamma}\mathcal{D}_{\gamma}(y_{1:n};\theta) \mid y_{1:n}\right],
\end{split}
\end{equation}
where $C^{(1)}$ and $C^{(2)}$ are defined in (\ref{C}).

\subsection{General case}\label{sec:dH-general}

Let $f(y_i;\theta)$ be a parametric density of $y_i$.
The density power divergence \citep{basu1998robust} is 
$$
\mathcal{D}_{\gamma}(y_i;\theta)=\frac1{\gamma}f(y_i;\theta)^{\gamma}-\frac1{1+\gamma}\int f(t;\theta)^{1+\gamma}dt,
$$
noting that the second term is irrelevant in the computation of the H-score since it does not depend on $y_i$.
The detailed expressions of the quantities that appear in the derivative of the H-score in (\ref{derive-H-score}) are obtained as follows:
\begin{align*}
&C_{\gamma}^{(1)}(y_i;\theta)=(\gamma-1)f(y_i;\theta)^{\gamma-2}f'(y_i;\theta)^2+f(y_i;\theta)^{\gamma-1}f''(y_i;\theta)+f(y_i)^{2(\gamma-1)}f'(y_i)^2\\
&C_{\gamma}^{(2)}(y_i;\theta)=f(y_i;\theta)^{\gamma-1}f'(y_i;\theta), \\
&\frac{\partial}{\partial\gamma}\mathcal{D}_{\gamma}(y_{1:n};\theta)
=\frac1{\gamma^2}\sum_{i=1}^nf(y_i;\theta)^{\gamma}\left\{\gamma \log f(y_i;\theta) - 1\right\}\\
& \ \ \ \ \ \ \ \ \ \ 
+\frac{n}{(1+\gamma)^2}\int f(t;\theta)^{1+\gamma}dt 
-\frac{n}{1+\gamma}\int f(t;\theta)^{1+\gamma}\log f(t;\theta)dt\\
&\frac{\partial}{\partial\gamma}C_{\gamma}^{(1)}(y_i;\theta)
=f(y_i;\theta)^{\gamma-2}f'(y_i;\theta)^2+\left\{C_{\gamma}^{(1)}+f(y_i;\theta)^{2(\gamma-1)}f'(y_i;\theta)\right\}\log f(y_i;\theta)\\
&\frac{\partial}{\partial\gamma}C_{\gamma}^{(2)}(y_i;\theta)
=C_{\gamma}^{(2)}(y_i;\theta)\log f(y_i;\theta)
\end{align*}

\subsection{Normal distribution case}
When $y_i\sim \mathcal{N}(\mu, \sigma^2)$, the corresponding density power divergence is 
$$
\mathcal{D}_{\gamma}(y_i;\theta)=\gamma^{-1}\phi(y_i;\mu,\sigma^2)^{\gamma}-(2\pi\sigma^2)^{-\gamma/2}{(1+\gamma)^{-3/2}}.
$$
The detailed expressions of quantities appeared in the derivative of the H-score in (\ref{derive-H-score}) are obtained as follows:
\begin{align*}
&C_{\gamma}^{(1)}(y_i;\theta)
=\frac1{\sigma^4}\left[w_i\left\{\gamma(y_i-\mu)^2-\sigma^2\right\}+w_i^2(y_i-\mu)^2\right], \ \ \ \ \ 
C_{\gamma}^{(2)}(y_i;\theta)= \frac{d}{dy_i}\mathcal{D}_{\gamma}(y_i;\theta)=-\frac{w_i(y_i-\mu)}{\sigma^2},\\
&\frac{d}{d\gamma}\mathcal{D}_{\gamma}(y_{1:n};\theta)
=
\frac1{\gamma^2}\sum_{i=1}^nw_i\left\{\gamma\log \phi(y_i;\mu, \sigma^2) -1\right\}
+ \frac{n}{2} (2\pi\sigma^2)^{-\gamma/2}(1+\gamma)^{-5/2}\left\{(1+\gamma)\log(2\pi\sigma^2)+3 \right\},\\
&\frac{\partial}{\partial\gamma}C_{\gamma}^{(1)}(y_i;\theta)
=\frac1{\sigma^4}\Big[w_i\left\{\gamma(y_i-\mu)^2-\sigma^2\right\}\log \phi(y_i;\mu, \sigma^2) + w_i(y_i-\mu)^2 + 2w_i^2(y_i-\mu)^2\log \phi(y_i;\mu, \sigma^2)\Big],\\
&\frac{\partial}{\partial\gamma}C_{\gamma}^{(2)}(y_i;\theta)
=
-\frac{w_i(y_i-\mu)}{\sigma^2}\log \phi(y_i;\mu, \sigma^2), 
\end{align*}
where $w_i=\phi(y_i;\mu, \sigma^2)^{\gamma}$. When $y_i\sim \mathcal{N}(x_i^\top\beta,\sigma^2)$, the derivative of the H score in the model is obtained by replacing $\mu$ with $x_i^\top\beta$.

\end{document}